\documentclass[12pt]{article}

\usepackage[colorlinks=true]{hyperref}
\usepackage{amsfonts,amsmath,amsthm}
\usepackage{amssymb}
\usepackage{hyperref}
\usepackage{graphicx}
\usepackage[top=1in,bottom=0.75in,left=0.514in,right=0.514in]{geometry}
\usepackage{enumerate}
\newtheorem{thm}{Theorem}
\newtheorem{lem}{Lemma}
\newtheorem{Def}{Definition}
\newtheorem{rem}{Remark}
\newtheorem{prop}{Proposition}
\newtheorem{coro}{Corollary}
\newtheorem{exam}{Example}

\newcommand{\tr}[1]{\text{Tr}\left(#1\right)}

\begin{document}

\title{Heisenberg Picture Approach to the Stability of Quantum Markov Systems%
\thanks{We gratefully acknowledge support by the Australian Research Council Centre of Excellence for Quantum Computation and Communication Technology (project number CE110001027),  Australian Research Council Discovery Project (project number DP110102322) and the Air Force Office of Scientific Research (grant numbers FA2386-09-1-4089 and FA2386-12-1-4075).}}

\author{Yu Pan$^\dagger$%
\thanks{$^\dagger$Research School of Engineering, Australian National University, Canberra, ACT 0200, Australia. {\tt\small
  \{yu.pan, zibo.miao\}@anu.edu.au}.}
\and  Hadis~Amini$^\ddagger$%
\thanks{$^\ddagger$Edward L. Ginzton Laboratory, Stanford University, Stanford, CA 94305, USA. {\tt nhamini@stanford.edu}.}
\and Zibo~Miao$^\dagger$
\and John~Gough$^*$
\thanks{$^*$Institute of Mathematics and Physics, Aberystwyth University
SY23 3BZ, Wales, UK. {\tt\small jug@aber.ac.uk}.}
\and Valery Ugrinovskii$^\sharp$%
\thanks{$^\sharp$School of Engineering and Information Technology,
  University of New South Wales at ADFA, Canberra,
ACT 2600, Australia. {\tt\small v.ugrinovskii@gmail.com}.}
\and Matthew~R.~James$^\S$%
\thanks{$^\S$ARC Centre for Quantum Computation and Communication Technology, Research School of Engineering, Australian National University,
Canberra, ACT 0200, Australia. {\tt\small matthew.james@anu.edu.au}.}}

\maketitle

\begin{abstract}
Quantum Markovian systems, modeled as unitary dilations in the quantum stochastic calculus of Hudson and Parthasarathy, have
become standard in current quantum technological applications. This paper investigates the stability theory of such systems.
Lyapunov-type conditions in the Heisenberg picture are derived in order to stabilize the evolution of system operators
as well as the underlying dynamics of the quantum states. In particular, using the quantum Markov semigroup associated with this
quantum stochastic differential equation, we derive sufficient conditions for the existence and stability of a unique and faithful invariant quantum
state. Furthermore, this paper proves the quantum invariance principle, which extends the LaSalle invariance principle to quantum systems in the
Heisenberg picture. These results are formulated in terms of algebraic constraints suitable for engineering quantum systems that are used in
coherent feedback networks.
\end{abstract}


\section{Introduction}\label{sec1}
The last two decades have witnessed a rapid development of quantum control
technologies, which have proved to be crucial in a large number of
quantum systems applications that require a high level of reliability, such as the
generation of on-demand quantum states, or the regulation of system performance
for quantum information processing
\cite{Wiseman09,Mabuchi05,Dong10,Altafini10,Zhang12}. Stability is central
to these quantum control systems. For example, quantum control tasks may
require stabilization of stochastic filtering process
\cite{mirra2007,Sayrin12,Handel05,Amini12}, of quantum oscillators in
optical systems \cite{Hamerly12,Serafini12}, or of a complex network
constructed via the coherent interconnection of quantum components
\cite{Yanagisawa03,James08,Nurdin09,ZhangJ12,Gough10}.

The traditional approach has been Schr\"{o}dinger picture Lyapunov techniques, that is, where
one defines a Lyapunov function as a positive function over the set of
states. This has lead to several important results on the stability of
quantum states in different control settings
\cite{BG07,mirra2007,wang2010,tico2013,Qi13}. In general, the main concern in
these control problems is the asymptotic behaviour of quantum systems in the
Schr\"{o}dinger picture.

In contrast, we wish to develop a Heisenberg picture Lyapunov approach which
exploits the fact that the Heisenberg picture more readily captures the physical dynamics,
and that this allows for a more direct extension of classical stability techniques.
We are interested primarily in open quantum systems and, in particular, the quantum stochastic calculus of Hudson and Parthasarathy \cite{HP1984}
gives the appropriate mathematical description. This moreover turns out to be
a very convenient set up for modeling open quantum coherent
networks. It has been shown in a number of references
\cite{Gough09,Zhang12} that evolution of such networks in the Heisenberg
picture can be associated with what is now sometimes referred to as the $(S,L,H)$ representation. (Here
$S$ is a scattering matrix and $L,H$ are coupling and Hamiltonian operators
of the system, respectively, and appear as coefficients in the quantum stochastic differential equation
for the unitary evolution process \cite{HP1984,Parth1992}. As we are interested in the average dynamics in
the vacuum state for the environment, we shall take $S=I$ for simplicity as only $L$ and $H$ appear in the
Lindbladian.) Ideally,
conclusions about stability properties of these operators should be made
according to their dynamics in the Heisenberg picture, particularly in the
form of conditioning on the infinitesimal generator of a Lyapunov operator
which will be introduced in Section \ref{sec3}. It is worth mentioning that the
Lyapunov operator is a generalization of the existing results about quantum
stability based on operator inequalities, which are closely related to the
dissipativity of the system \cite{JG10,PUJ1}. To fix ideas, suppose that the evolution of an observable $X\in \mathfrak{A}$ in the Heisenberg picture
is given by $:t \mapsto j_t (X)$ and that this may be described by a Langevin equation
\begin{equation*}
dj_t(X) = j_t( \mathcal{G} (X)) dt + \mathrm{Noise},
\end{equation*}
where $\mathcal{G}$ is a Lindblad generator and the ``Noise'' terms are martingale increments for the environment state.
Mirroring the approach to the classical stochastic stability
theory, the dynamics of the expectation of an operator can then be
established from the properties of the corresponding infinitesimal
generator of a Lyapunov operator $V$ and, for instance, an exponential stability criterion would be that
\begin{equation}\label{Gv}
\mathcal{G}(V)\leq -cV+dI,\quad c,d>0,
\end{equation}
see e.g., \cite{JG10,PUJ1,PUJ4a}. Here $I$ is the identity
operator, and we include a dissipation rate $d$. We point out a similarity
between (\ref{Gv}) and conditions that arise in the Lyapunov stability
theory of classical stochastic differential equations~\cite{Khas}.

Our aim in this paper is to further develop the Lyapunov stability theory
of quantum stochastic evolutions in the Heisenberg picture. First, we will
study an interplay between stability of states and stability of
operators. For example, we will employ the Lyapunov condition~(\ref{Gv}) to
infer the asymptotic behavior of quantum states from the corresponding
properties of
the operator $V$. In particular, we will present a stability analysis of
the invariant state of an open quantum system based on the Lyapunov method
and the quantum semigroup theory. In a sense, this contribution is in
parallel with the results in the classical stochastic stability theory,
such as the Foster-Lyapunov theory~\cite{MT-1993}, concerned with the
existence and stability of invariant probability measures of Markov
processes.

In addition to invariant states, we are also interested in invariant
sets of states or operators to which quantum evolutions converge. To
characterize the invariance property of the system, we develop a quantum
version of the LaSalle invariance principle in the Heisenberg
picture. Using this result, we are able to determine a limit set of state
trajectories, as well as explore the possibility of stabilizing two
non-commuting operators simultaneously. Similar conditions that employ
Lyapunov techniques have been developed for classical stochastic systems
\cite{Khas,Kushner67,Mao99}. The objective of our development is to pave the
way to the design of coherent feedback networks, as operators of a
quantum system form a non-commutative algebra. An alternative approach of
analyzing stability of state trajectories in the Schr\"{o}dinger picture is
often too difficult in this case.

The main results of this paper are formulated in the way such that the ground states of an operator $V$ or $W$ are stabilized, given that certain Lyapunov conditions are satisfied. This is directly relevant to a recent quantum information processing scheme \cite{Verstraete09}, where the task is either stabilizing the ground states of a Hamiltonian which encode the solution to a quantum computation problem, or robustly preparing entangled quantum states \cite{Verstraete09,Lin13} through engineering the dissipative property of the systems. The results obtained in this paper thus provide tools to design proper environmental couplings for these applications.

This paper is organized as follows: In Section \ref{sec2}, we briefly
review some facts about Markov dynamics of open quantum systems. In Section
\ref{sec3}, the definition of the Lyapunov operator is given. Then in
Section \ref{sec4}, we propose a Lyapunov condition to ensure the existence of an
invariant state. In Section \ref{sec5}, we prove the faithfulness and uniqueness of
an invariant state using certain non-degeneracy conditions imposed on
the diffusion coefficients of the quantum stochastic differential
equation. In Section \ref{sec6}, we derive the quantum invariance principle and
discuss its implications. In Section \ref{sec7}, we investigate the stability
within the invariant set and provide a sufficient condition to stabilize
the system to the ground state. The last section is devoted to conclusions.

\emph{Notations}. In this paper $\mathfrak{H}$ is a separable Hilbert space, and $\mathfrak{A}$
denotes the von Neumann algebra of operators acting on the Hilbert space $\mathfrak{H}$. The commutator of two operators $A$ and $B$ is written as
$[A,B]=AB-BA$, while $X^\dagger$ is the adjoint of an operator
$X$. $\mathfrak{M}^{'}=\{A\in \mathfrak{A} :[A,X]=0, X\in \mathfrak{ M} \}$ denotes the commutant of a von Neumann algebra
$\mathfrak{M}$ of operators. In particular we write $\mathbb{C}I$ for the trivial von Neumann algebra consisting of
multiples of identity. $\mathbb{C}$ is the set of complex numbers. A positive-semidefinite operator $X$ is indicated as $X\geq0$. We shall only work with normal states, and
typically write $\rho$ for the corresponding density operator so that the expectation of an
operator $X$ will be denoted as $\langle
X\rangle_{\rho}=\tr{\rho X}$, e.g. \cite{Sakurai1985}.

\section{Quantum Markov Systems}\label{sec2}
Consider the system defined on a Hilbert space $\mathfrak{H}_S$, and the
environment on a Fock space $\mathfrak{H}_B$ over $L^2 (\mathbb{R}_+ , dt ) $ corresponding to a single Boson field mode. The composite system can be
regarded as a single closed system whose dynamics are
characterized by a unitary evolution $U(t)$ on $\mathfrak{H}_S\otimes \mathfrak{H}_B$ obeying a quantum stochastic differential
equation \cite{HP1984,JG10}
\begin{equation*}
dU(t)=\{dB^\dagger L-L^{\dagger}dB-\frac{1}{2}L^\dagger Ldt-iHdt\}U(t).
\end{equation*}
Here $H$ is the Hamiltonian of the system, and $L$ describes the
coupling between the system and environment; $B(\cdot )$ and $B^\dagger (\cdot )$ are
the annihilation and creation process defined on $\mathfrak{H}_B$. In the Heisenberg picture an operator $X(t)=j_t(X)$ of the system evolves as $j_t(X)=U(t)^\dagger (X\otimes
I)U(t)$. Given the interaction Hamiltonian of the
combined system, the explicit dynamical equation for $X(t)$ can be
written as
\begin{equation}~\label{eq:dynone}
dj_t(X)=j_t(\mathcal{G}(X))dt+ j_t (\mathcal{B}(X))dW_1(t) + j_t (\mathcal{C}(X))dW_2(t).
\end{equation}
Here the following notation is used
\begin{eqnarray}~\label{eq:generator}
\mathcal{G}(X) &=& -i[X,H]+\mathfrak{L}(X),\\ \notag
\mathcal{B}(X) &=& \frac{1}{2}([X,L]+[L(t)^\dagger,X]),\\ \notag
\mathcal{C}(X) &=& \frac{i}{2}(-[X,L]+[L^\dagger,X]).
\end{eqnarray}
with
\begin{equation}~\label{eq:Lindblad}
\mathfrak{L} (X)= L^\dagger{X}L-\frac{1}{2}L^\dagger{L}X-\frac{1}{2}XL^\dagger{L}.
\end{equation}
We have written the noise increments in quadrature form, that is, in terms of
\begin{equation*}
dW_1(t)=dB(t)+dB^\dagger(t), \quad dW_2(t)=i(dB(t)^\dagger-dB(t)),
\end{equation*}
with all increments understood in the It\={o} sense.

Alternatively, we can characterize the average evolution of $X(t)$ by the semigroup $T_t$ acting as
\begin{equation*}
T_t(X)=\mathbb E_0(U(t)^\dagger(X\otimes I)U(t)).
\end{equation*}
$\mathbb
E_0$ is a conditional expectation on the given initial algebra $\mathfrak{A}$ and the
initial vacuum state $|0\rangle\langle0|$ of $\mathfrak{H}_B$. The infinitesimal generator of this Markov semigroup is then given
by $\mathcal{G}$. The dissipation functional of the semigroup is defined as \cite{Lind1976,Frigerio78}
\begin{equation}
\mathfrak{D}(X)=\mathcal{G}(X^\dagger X)-\mathcal{G}(X^\dagger)X-X^\dagger\mathcal{G}(X).
\label{dissip.functional}
\end{equation}
For a completely positive semigroup $T_t$, we have $\mathfrak{D}(X)\geq0,X\in\mathfrak{A}$. The dissipation functional characterizes the irreversible nature of the quantum Markov process, and consequently the system is dissipative if $\mathfrak{D}\neq0$.

The corresponding semigroup in the predual space of the trace
class operators (the state space) is
denoted as $T_{*t}(\rho)=\rho_t$. The support projection $P_\mathrm{sup}$ of a state $\rho$ is defined as the smallest projection (in the sense that $P_\mathrm{sup}\le P$) to which the state assigns probability 1.

\begin{Def}[\cite{fagnola2003quantum}]
A state $\rho_I$ is an invariant state, if it satisfies the condition
$T_{*t}(\rho_I)=\rho_I$ $\forall t\ge 0$. A state $\rho$ is faithful in $\mathfrak{A}$ if $\tr{\rho A}=0$ implies $A=0$ for any
positive operator $A\in\mathfrak{A}$. In other words, $\rho$ is faithful if the support
projection of $\rho$ is the identity operator in the space of bounded operators on the underlying Hilbert space.
\end{Def}

The faithfulness of an invariant state is essential to the analysis of the
asymptotic
behaviour of a quantum Markov system. For example, if the system
possesses a faithful invariant state, then its ergodic properties and the
problem of convergence to equilibrium can be studied for this system
\cite{Alberto1982}. Moreover, if this faithful invariant state is unique,
then it is the only equilibrium state of the system \cite{Fagnola2004}.

\section{Quantum Lyapunov Operators and Stability in the Heisenberg Picture}\label{sec3}
In classical theory, stability refers to the property that the trajectories of the dynamical systems will remain near an equilibrium point $x_e$ if the initial states $x_0$ are near $x_e$. A stronger notion is asymptotic stability which additionally requires that the trajectories that start near the equilibrium state $x_e$ will converge to $x_e$. In practice it
is often too complicated for nonlinear systems to solve the dynamical equations directly, and the main tool for proving stability of
these systems is Lyapunov theory \cite{Khas,Kushner67,Thygesen97asurvey} without finding the trajectories. Generally
speaking, if a given system possesses a Lyapunov function $V(x)$,
with certain conditions on $V(x)$ and its convective derivative $\dot V(x)$,
then the trajectories of the system state $x_t$ will be stable in some
sense. For example, any continuous scalar function $V:\mathbb{R}^n\rightarrow
\mathbb{R}$ having the property
\begin{equation*}
V(0)=0,\quad V(x)>0,\quad \dot{V}(x)<0,\quad x\in\mathbb{R}^n\backslash\{0\}
\end{equation*}
can be chosen as a Lyapunov function for the purpose of establishing
asymptotic stability of the zero equilibrium state of the
system. Alternatively, the Lyapunov function can be defined as a continuous
function $V(x)$ satisfying \cite{Khas,Mali09}
\begin{equation*}
a(|x|)\leq V(x)\leq b(|x|), \quad \dot{V}(x)\leq 0,
\end{equation*}
where $a(\cdot),b(\cdot)$ are strictly increasing functions with
$a(0)=b(0)=0$ and $a(\infty)=b(\infty)=\infty$. In both cases, if such a
$V$ exists, then any trajectory $x_t$ will converge to $0$.

As noted in the introduction, stability of quantum systems
may be considered within either the Schr\"{o}dinger or Heisenberg pictures. In this paper, our intention is
to develop the Heisenberg picture approach and tools for studying stability
of state trajectories $\rho_t$ within the underlying Schr\"{o}dinger
picture. To this end, we define the Lyapunov operator in the Heisenberg
picture.
\begin{Def}
A quantum Lyapunov operator $V$ is an observable (self-adjoint operator) on a Hilbert space $\mathfrak{H}$ for
which the following properties hold:
\begin{enumerate}
\item $V\in D(\mathcal G)$,
\item  $V\geq 0$,
\item $\mathcal{G}(V)\leq 0$.
\end{enumerate}
\end{Def}
$D(\mathcal G)$ is the domain of the generator. Note that $\langle V\rangle_\rho\geq0$ for all states $\rho$. In the next two sections, we will show that the existence and stability of the invariant state of the system can be proved using a Lyapunov operator.

One advantage of the Heisenberg approach is that the stability of operators in the
Heisenberg picture may be studied, and not just stability of states. This is of
practical importance, especially within the framework of
quantum coherent networks. Although the stability of operators has been
studied using operator semigroup theory, to the best of our knowledge the
first approach to stabilization of quantum systems via Lyapunov methods is in \cite{JG10}.

We now introduce our concept of stability of operators in the Heisenberg picture.
\begin{Def}
Given a set of positive operators $\mathfrak{S}$, the system is $\mathfrak{S}$-stable if $\langle j_t(A) \rangle_{\rho}$ is bounded for each $A \in \mathfrak{S}$ and any initial state $\rho$.
\end{Def}
In our development of Lyapunov quantum stability, we will make use of the notion of quantum coercivity defined in terms of the spectral decomposition of a Lyapunov operator $V$.
\begin{Def}
Consider a positive operator $V$ with the spectral decomposition
$V=\sum_iv_iP_i$, $v_i$ being the eigenvalues of $V$. $V$ is coercive if
there exists a strictly increasing function $k(\cdot)$ with
$\lim_{i\to \infty}k(i)=\infty$ such that
$v_i\geq k(i), i>i_0$, for  some $i_0$.
\end{Def}
Note that the definition of quantum coercivity is analogous to its
classical counterpart \cite{Renardy04,MT-1993}
\begin{equation*}
V(x)\geq c(|x|)
\end{equation*}
with $c(|x|)$ being a strictly increasing function that goes to infinity as
$|x|\rightarrow\infty$.

From Definition~$4$, it follows that if the system possesses a coercive
Lyapunov operator $V=\sum_iv_iP_i$, then the set of operators
$\mathfrak{S}=\{A\geq0:\tr{AP_i}\leq \epsilon k(i),\epsilon>0\}$ are bounded in
expectation, hence the system is $\mathfrak{S}$-stable. Also, in Section \ref{sec6} we will prove that when given a Lyapunov
operator $V$, the set $\mathfrak{S}$ of operators defined by $\mathfrak{S}=\{W\geq0:\mathcal{G} (V)\leq -W\}$ are bounded in expectation. Indeed, the expectation $\langle
W(t)\rangle_{\rho}$ will converge to zero according to quantum LaSalle
invariance principle. Hence, a conclusion about stability of the system can
be made.

As in the classical case, one may have a number of variations on the definition of a Lyapunov function,
depending on the context of stability property which one is interested
in. The definition of Lyapunov operator can be relaxed for
quantum stability analysis in different contexts. Therefore, we still call
$V$ a Lyapunov operator when the property $\mathcal{G}(V)\leq0$ is replaced
by a weaker condition (\ref{Gv}), as in the following example.

\begin{exam}\rm
Consider a quantum oscillator with the Hamiltonian given by
$H={\omega}a^\dagger{a}$, and the coupling operator $L=
\alpha{a}+\beta{a^\dagger}$, $a$ and $a^\dagger$ are annihilation and creation operators respectively, and they satisfy the commutation relation $[a,a^\dagger]=1$. Choose the candidate Lyapunov operator as the photon number operator $V=a^\dagger{a}$ which represents the energy of the system. By calculation we find $\mathcal{G}(V)=-(|\alpha|^2-|\beta|^2)V+|\beta|^2I$. If $|\alpha|^2>|\beta|^2$, $\mathcal{G}(V)$ satisfies the condition (\ref{Gv}) and $V$ becomes a Lyapunov operator in this problem. Furthermore, the system is $\mathfrak{S}$-stable with $\mathfrak{S}$ being the von Neumann algebra generated by $V$ since $\langle V(t)\rangle$ is bounded \cite{JG10}. If $|\alpha|^2<|\beta|^2$, $\langle V(t)\rangle$ is unbounded and the system is unstable in energy.

In the sequel, a Lyapunov operator $V$ for which condition (\ref{Gv}) holds is referred to as a quantum Lyapunov operator in the weak sense.
\end{exam}

\section{Quantum Tightness and the Existence of Invariant States}\label{sec4}
As a first step to study the stability of quantum states, we derive certain conditions to guarantee the existence of
invariant state.
First we present the definition of quantum tightness \cite{Meyer1995}.
\begin{Def}
A sequence $(\rho_n)_{n\geq 1}$ in the Banach space of trace-class
operators on a Hilbert space $\mathfrak{H}$ is tight if for every $\epsilon>0,$ there exists a
finite rank projection $P$ and $n_0>0$ such that
$\tr{\rho_nP}>1-\epsilon$ for all $n\geq n_0.$
\end{Def}
Obviously, trajectories of states corresponding to finite-dimensional
systems are tight. We will refer to the following lemma. \cite{Meyer1995}
\begin{lem}
A tight sequence $(\rho_n)_{n\geq 1}$ of quantum states admits a
subsequence converging to a quantum state.
\end{lem}

\begin{thm}[\cite{fagnola2003quantum}]
If the system possesses a tight family of quantum states,
$(\rho_t,t>0)$ then the system possesses at least one invariant state.
\end{thm}
\begin{proof} As $\rho_t$ is tight, any sequence of states
$\rho_{t_n}=\frac{1}{t_n}\int_0^{t_n}\rho_{t^{'}}dt^{'}$ is also tight and therefore has normalized
sequential limit points. These states are invariant because any sequential limit point of
$\frac{1}{t_n}\int_0^{t_n}\rho_{t^{'}}dt^{'}$ is invariant, according to
Proposition~$2.3$ in \cite{fagnola2003quantum}.
\end{proof}

Based on these properties, we can develop the condition on the tightness of
general quantum systems. Recall the inequality (\ref{Gv}) involving the (Lindblad)
generator $\mathcal{G}$ of a quantum Markov process
defined by Equation~\eqref{eq:generator}. Suppose $V$ is a Lyapunov operator in the weak sense, i.e., $\mathcal{G}(V)\leq -cV+dI, c>0$. By integrating  (\ref{Gv})
we obtain the following inequality \cite{JG10}
\begin{equation*}
\langle V(t)\rangle\leq e^{-ct}\langle V(0)\rangle+\frac{d}{c},
\end{equation*}
which means $\langle V(t)\rangle\leq\lambda$ for any $t\geq0$ and some
positive $\lambda$. Next we will show that the condition (\ref{Gv}) not only
gives us the mean stability of $V$ but also implies tightness of the
corresponding collection of quantum states $\{\rho_t, t\ge 0\}$.

First, let us consider the following example.
\begin{exam}\rm
The photon number operator for a quantum oscillator can be written as
$V=\sum_0^{\infty}i|i\rangle\langle{i}|$ where $|i\rangle$ is the
photon number state. If $\langle V(t)\rangle\leq{c},c\geq0$, then we have
$\sum_{i=0}^{\infty}i\rho_t^{ii}\leq{c}$ for an
  arbitrary sequence of states $\rho_{t}$; here
  $\rho_t^{ii}=\tr{\rho_t|i\rangle\langle{i}|}$.
 For an
    arbitrary $\epsilon>0$, choose $m$
such that $m=\left[\frac{c}{\epsilon}\right]$, where $[x]$ is the nearest
  integer to $x$ that is greater than $x$. Through
$m\sum_{i=m}^{\infty}\rho_t^{ii}\le\sum_{i=m}^{\infty}i\rho_t^{ii}\leq{c}$, we conclude $\sum_m^{\infty}\rho_t^{ii}\le\epsilon$ for any $t$. The
finite rank projection $P=\sum_{i=0}^{m}|i\rangle\langle{i}|$
then satisfies the condition $\tr{\rho_{t}P}>1-\epsilon$, which indicates
that the sequence $\rho_{t}$ is tight. Hence the corresponding state
trajectory of the quantum oscillator gives rise to an invariant state for
the oscillator.
\end{exam}

The example shows that under certain conditions, the stability of an
operator in the mean sense may
imply tightness of a corresponding state trajectory.
The inequality
$$m\sum_m^{\infty}\rho_t^{ii}\le\sum_m^{\infty}i\rho_t^{ii}\leq{c}$$
is essential in this example. In fact, the spectral property
of the above operator is the key element connecting tightness and
stability. We generalize this idea in the following theorem.
\begin{thm}
Suppose the evolution of a positive observable $V$ on a separable Hilbert
space $\mathfrak{H}$, with spectral decomposition as $V=\sum_{i=0}^{\infty}v_iP_i$, is stable
 in the mean, that is, there exists a constant $c\ge0$ such that $\langle V(t)\rangle_\rho\leq{c}$ with $\rho$ as the initial state. If $V$ is coercive, then any sequence $\rho_t$ is tight which implies the existence of an invariant state.
\end{thm}
\begin{proof}
The proof is similar to the proof used in Example~$2$ to show
tightness. The condition that $\langle V(t)\rangle_\rho\leq{c}$ means
$\sum_{i=0}^{\infty}v_{i}\rho_t^{ii}\leq{c}$ for $t\geq0$. Here $\rho^{ii}=\tr{\rho P_i}$
denotes the projection $P_i$ on the state $\rho$. Since $V$ is coercive, there exists some
$N_0$ such that $v_i$ is increasing for $i\geq N_0$ and
$v_i\rightarrow\infty$ as $i\rightarrow\infty$. Choose
$m=\max\{N_0,\inf\{i:v_{i}\geq\frac{c}{\epsilon}\}\}$. Then
$v_{m}\sum_m^{\infty}\rho_t^{ii}\le\sum_m^{\infty}v_{i}\rho_t^{ii}\leq{c},$
so we find that $\sum_{i=m}^{\infty}\rho_t^{ii}\le\epsilon$. Letting
  $P=\sum_{i=0}^{m-1}P_i$, we obtain $\tr{\rho_tP}>1-\epsilon$; i.e, $\rho_t$ is tight. The result of the theorem then follows from Theorem~1.
\end{proof}

It follows from Theorem~$2$ that the existence of a coercive Lyapunov operator in the weak sense (\ref{Gv}) guarantees the existence of an invariant state. This prompts the question as to under what condition such an invariant state is unique and/or faithful. This question is addressed in the next section.

\section{Stability of Invariant States}\label{sec5}
In this section, we obtain some conditions to guarantee the
faithfulness and uniqueness of an invariant state.

For a particular invariant state $\rho_I$, its support
projection is denoted as $P_I$.
We shall need the following proposition.
\begin{prop}[see e.g.,~\cite{fagnola2003quantum}]
The support projection of an invariant state is subharmonic. That is, $T_t(P_I)\geq P_I$.
\end{prop}
The above property of the support projection can be expressed
in terms of the generator $\mathcal{G}$ of the semigroup $T_t$ as
$\mathcal{G}(P_I)\geq0$.

\medskip

\subsection{Stability of invariant states of finite-dimensional systems}
For a finite-dimensional system with the underlying Hilbert space $\mathfrak{H} =\mathbb{C}^n$, the following
  theorem determines faithfulness and uniqueness of an invariant
state.
\begin{Def}
A state $\rho$ is said to be globally attractive if all system trajectories asymptotically converge to $\rho$ for any initial state.
\end{Def}
\begin{thm}~\label{thm:three}
Suppose $\mathfrak{H}=\mathbb{C}^n$. If $PL^\dagger(I-P)LP\neq0$ for any non-trivial projection $P$, then the invariant state $\rho_I$ is faithful and unique.
\end{thm}
\begin{proof}
A finite dimensional system is tight by Definition~$5$ and therefore, according to Theorem 1, it admits an invariant
state $\rho_I$.  Let $P_I$ be the support projection of $\rho_I$.

If we take any orthogonal projection $P$, then we have
\begin{equation*}
\mathcal{G}(P)=\mathcal{G}(P^2)=P\mathcal{G}(P)+ \mathcal{G}(P)P +\mathfrak{D} (P),
\end{equation*}
where $\mathfrak{D}$ is the dissipation functional defined in (\ref{dissip.functional}), and so
\begin{equation*}
P\mathcal{G}(P)P=-P\mathfrak{D}(P)P.
\end{equation*}
However we note that $\mathfrak{D}(X) = [X,L]^\dag [X,L] \geq 0$, and in particular,
\begin{equation*}
P \mathfrak{D}(P) P = P L^\dag (I-P) L P.
\end{equation*}

Now take the invariant state $\rho_I$ with support projection $P_I$, then from Proposition~$1$ we will have $ \mathcal{G} (P_I) \geq 0$,
and therefore $P_I \mathcal G(P_I) P_I \geq 0$. But we then must have $P_I \mathfrak{D} (P_I) P_I = 0$, as $\mathfrak{D} \geq 0$.
We thereby deduce that for the invariant state support
\begin{equation}~\label{eq:P_I}
P_I \mathfrak{D}(P_I) P_I = P_I L^\dag (I-P_I) L P_I = 0.
\end{equation}
This is automatically satisfied if $\rho_I$ is faithful, since here $I - P_I \equiv 0$.

Suppose the hypothesis of the theorem is true, namely that $ P L^\dag (I-P) L P\neq0$ for any non-trivial orthogonal projection $P$. If we also now suppose
that $\rho_I$ is not faithful, then $P_I$ is non-trivial, then setting $P=P_I$ in (\ref{eq:P_I}) leads to a contradiction. Therefore, under the hypothesis,
we see that any invariant state must be faithful.

Suppose the invariant state is not unique, then there exist non-trivial orthogonal invariant subspaces by \cite{Bau12,Schirmer10}. This leads to a contradiction since there will exist non-faithful invariant states in the subspace.
\end{proof}

\begin{rem}
Condition $PL^\dagger(I-P)LP\neq0$ means that any non-trivial projection $P$ is connected with its orthogonal complement by $L$. This property can be easily verified when the system has reduced dynamics. For example, if the quantum states maintain a diagonal form $\rho(t)=\sum_i\rho^{ii}(t)P_i,P_i=|i\rangle\langle i|$ during evolution, we only need to verify $P_iL^\dagger(I-P_i)LP_i\neq0$ for all $P_i$. To generalize, if there exists a family of projections $\{P_i\}$ such that $\sum_iP_i=I$ and $P_iL^\dagger(I-P_i)LP_i\neq0$, the marginal distribution of the invariant state will have non-vanishing probability on each projector $P_i$.
\end{rem}

\medskip

It is worth mentioning that for finite dimensional system, uniqueness of invariant state directly leads to global convergence \cite{Schirmer10}.

\begin{exam}\rm
Consider the quantum two-level system with a basis denoted as $\{|0\rangle,|1\rangle\}$. $H=\omega\sigma_z=\omega(|0\rangle\langle0|-|1\rangle\langle1|)$ and $L=\sigma_x=|0\rangle\langle1|+|1\rangle\langle0|$. The quantum state evolves according to the master equation
\begin{equation*}
\frac{d\rho(t)}{dt}=-i[H,\rho(t)]+L\rho(t)L^\dagger-\frac{1}{2}L^\dagger L\rho(t)-\frac{1}{2}\rho(t)L^\dagger L.
\end{equation*}
Obviously the density matrix of the state will remain diagonal if the
initial state is $\alpha|0\rangle\langle0|+\beta|1\rangle\langle1|$ with
arbitrary $\alpha$ and $\beta$ satisfying $|\alpha|^2+|\beta|^2=1$. As a
result, the system will possess a diagonal invariant state. We only need to consider the
projections $\{|1\rangle\langle1|,|0\rangle\langle0|\}$ in order to
conclude faithfulness of this invariant state. In fact, we have
$|1\rangle\langle1|\sigma_x|0\rangle\langle0|\sigma_x|1\rangle\langle1|=|1\rangle\langle1|\neq0$
and
$|0\rangle\langle0|\sigma_x|1\rangle\langle1|\sigma_x|0\rangle\langle0|=|0\rangle\langle0|\neq0$,
so the two-level system has a unique faithful invariant state which is globally
attractive.
\end{exam}

\medskip

\subsection{Stability of invariant states of infinite-dimensional systems}
Now we can prove the main result in this section for the quantum system defined on a separable Hilbert space $\mathfrak{H}$.
\begin{thm}\label{theorem4}
Suppose there exists a coercive Lyapunov operator in the weak sense (\ref{Gv}). If $PL^\dagger(I-P)LP\neq0$ for any non-trivial projection $P$, then any invariant state $\rho_I$ is faithful and unique. Furthermore, this faithful state $\rho_I$ is globally attractive.
\end{thm}
\begin{proof}
The proof follows along the same lines as the proof of
Theorem~$3$. However, the existence of invariant state comes from condition
(\ref{Gv}) and coercivity, and is a direct result of
Theorem~$2$. The Lyapunov operator inequality (\ref{Gv}) and the algebraic
condition $PL^\dagger(I-P)LP\neq0$ are then combined to guarantee the uniqueness and faithfulness of this invariant state which is also the
equilibrium point of the system. In addition, the unique invariant state is also globally attractive due to its faithfulness
\cite{Fagnola2004}.
\end{proof}

We also note a certain analogy between Theorem~4 and the corresponding
results from the classical theory of stochastic Markov processes; e.g.,
see~\cite{MT-1993}.  In particular, our condition (\ref{Gv}) is analogous
to the positive recurrence condition (CD2) in~\cite{MT-1993}.

\medskip

\begin{exam}\rm
Consider again a quantum oscillator with the Hamiltonian $H={\omega}a^\dagger{a}$, and the coupling operator $L=
\alpha{a}+\beta{a^\dagger}$.

Consider the observable $V=a^\dagger{a}$ which has a strictly increasing and unbounded spectrum. $\mathcal{G}(V)=-(|\alpha|^2-|\beta|^2)V+|\beta|^2I$. In order to satisfy the Lyapunov
condition in Theorem~$4$, we need to set $|\alpha|>|\beta|$. In this case,
$\langle V(t)\rangle$ is bounded with respect to any initial state. Hence, according to
Theorem~2, this system admits an invariant state. Now we want to study the set of projections $\{P_i=|i\rangle\langle i|\}$. Note that since $|i\rangle$ is the photon number state, then $\sum_iP_i=I$. We have $|i\rangle\langle i|L^\dagger|i+1\rangle\langle i+1|L|i\rangle\langle i|=(i+1)|\beta|^2|i\rangle\langle i|\neq0$ and $|i\rangle\langle i|L^\dagger|i-1\rangle\langle i-1|L|i\rangle\langle i|=i|\alpha|^2|i\rangle\langle i|\neq0$. Therefore, any photon number state is connected to its two neighboring states, so by induction any non-trivial projection $P=\sum_jP_j\neq I$ is not a support projection of an invariant state $\rho_I$. Consequently, the photon-number distribution of the invariant state has non-vanishing probability on the
entire Fock basis. However, this does not imply that the invariant state is faithful because there may exist other set of projections that does not satisfy the algebraic condition of Theorem \ref{theorem4}.
\end{exam}

\section{Quantum LaSalle Invariance Principle}\label{sec6}
In the previous section, we studied the stability property of convergence
to faithful invariant
states. Other classes of stabilization problems of interest are concerned with
stability of non-commuting operators, or
require convergence to an invariant set for any state trajectories. Similar
to the classical LaSalle's invariance principle \cite{Lasalle68,Mao99} that
is used to identify the asymptotic stability of system trajectories, the
invariance theorems which we will derive here pave the way for analyzing
the underlying dynamics of general quantum states which may not be faithful in the Heisenberg picture.

The classical LaSalle theorem states the following fact \cite{Lasalle68}:
If a positive and uniformly continuous function $V(x)$ can be found on a
compact space such that $\dot V(x)\leq0$, then the limit points of any
trajectory $x_t$ are contained in the largest invariant subset of
$\{x:\dot{V}(x)=0\}$.

First we will derive the direct analogue of the classical LaSalle invariance theorem in the Heisenberg picture.
\begin{Def}
A quantum state $\rho$ is said to be the zero solution of an operator $X$ if $\rho$ solves $\langle X\rangle_\rho=0$.
\end{Def}
\begin{thm}
If there exists a coercive Lyapunov operator $V$ and a positive operator $W$ with $\mathcal{G}(W)$ bounded in the operator norm such that
\begin{equation}
\mathcal{G}(V)\leq-W,
\label{ineqlasa}
\end{equation}
then $\lim_{t\rightarrow\infty}\langle V(t)\rangle_{\rho_0}=\lim_{t\rightarrow\infty}\langle V\rangle_{\rho_t}$ exists for any initial state $\rho_0$ and
\begin{eqnarray}
&\int_0^\infty\langle W(t^{'})\rangle_{\rho_0}dt^{'}=\int_0^\infty\langle W\rangle_{\rho_{t^{'}}}dt^{'}<+\infty,&\nonumber\\
&\lim_{t\rightarrow\infty}\langle
W(t)\rangle_{\rho_0}=\lim_{t\rightarrow\infty}\langle W\rangle_{\rho_t}=0.&
\label{conlasa}
\end{eqnarray}
\end{thm}

\begin{proof}
Referring to Theorem 2, tightness of $\rho_t$ ensures the existence of a limit point (or an accumulation point) of the system evolutions. The function $\langle V(t)\rangle_{\rho_0}$ is decreasing with $t$ since $\mathcal{G}(V)\leq0$. Therefore, $\lim_{t\rightarrow\infty}\langle V(t)\rangle_{\rho_0}=\lim_{t\rightarrow\infty}\langle V\rangle_{\rho_t}$ exists because any decreasing sequence with a lower bound will converge to a limit.
Moreover, $\langle V(t)\rangle_{\rho_0}$ evolves according to
\begin{equation}
\langle V(t)\rangle_{\rho_0}-\langle
V\rangle_{\rho_0}=\int_0^t\langle\mathcal{G} (V(t^{'}))\rangle_{\rho_0}dt^{'}\leq\int_0^t\langle-W(t^{'})\rangle_{\rho_0}dt^{'}.
\label{integene1}
\end{equation}
It follows from (\ref{integene1}) that $\int_0^t\langle
W(t^{'})\rangle_{\rho_0}dt^{'}\leq\langle V\rangle_{\rho_0}$, which implies
\begin{equation*}
\int_0^\infty\langle W(t^{'})\rangle_{\rho_0}dt^{'}=\int_0^\infty\langle W\rangle_{\rho_{t^{'}}}dt^{'}<+\infty.
\end{equation*}
$\rho_t$ is tight, so the positive sequence $\langle W\rangle_{\rho_t}$
must have convergent subsequence. Suppose there exists a subsequence
$\rho_{t_k}$ such that $\lim_{k\rightarrow\infty}\langle
W\rangle_{\rho_{t_k}}=\epsilon>0$. Now we show that this leads to a
contradiction. Since $\mathcal{G}(W)$ is bounded in operator norm by $R$, we
have
\begin{equation*}
|\langle W(t_1)\rangle_{\rho_0}-\langle W(t_2)\rangle_{\rho_0}|=|\int_{t_2}^{t_1}\tr{\mathcal{G}(W)\rho_{t^{'}}}dt^{'}|\leq\int_{t_2}^{t_1}\|\mathcal{G}(W)\||\rho_{t^{'}}|dt^{'}\leq R|t_1-t_2|,
\end{equation*}
which means $\langle W\rangle_{\rho_t}$ is uniformly continuous in $t$. Here $|\rho_{t^{'}}|=1$ denotes the trace-norm of the density state $\rho_{t^{'}}$. According to the uniform continuity, we are able to find a $\delta>0$ such that the following inequality
\begin{equation*}
\langle W\rangle_{\rho_{t^{'}}}>\frac{\epsilon}{2}
\end{equation*}
holds if $|t^{'}-t_k|<\frac{\delta}{2}$ for any $t_k$. This further implies
\begin{equation*}
\int_0^\infty\langle W\rangle_{\rho_t}dt\geq\sum_k^\infty\delta\frac{\epsilon}{2}=+\infty,
\end{equation*}
which is a contradiction.

The above contradiction implies that every converging subsequence of
$\langle W\rangle_{\rho_t}$ converges to 0. Then we conclude that
$\lim_{t\rightarrow\infty}\langle W\rangle_{\rho_t}=0$.
\end{proof}

\begin{rem}
For a Lyapunov operator $V$ with $\mathcal{G}(V)\leq0$, we can always let $W=-\mathcal{G}(V)$ and thus the trajectories will converge to $\{\rho:\langle\mathcal{G}(V)\rangle_\rho=0\}$ if $\mathcal{G}(\mathcal{G}(V))$ is bounded, according to Theorem~$5$. The states from the invariant set $\{\rho:\langle\mathcal{G}(V)\rangle_\rho=0\}$ are zero solutions of $W$. This conclusion is similar to the statement of the classical LaSalle theorem.
\end{rem}

\begin{coro}
If Inequality (\ref{ineqlasa}) in Theorem~$5$ is replaced by
\begin{equation*}
\mathcal{G}(V)\leq U-W,
\end{equation*}
where $U$ is a positive operator satisfying
\begin{equation*}
\int_0^\infty\langle U(t)\rangle_{\rho_0}dt<\infty
\end{equation*}
for any initial state $\rho_0$, the conclusions of Theorem~$5$ still hold.
\end{coro}
\begin{proof}
The proof is similar to the proof of Theorem~$5$.
\end{proof}

The question is how we can characterize the pairs of operators $V$ and $W$ for which (\ref{ineqlasa}) holds. Note that $\langle\mathcal{G}(V)\rangle_{\rho}\geq0$ for any ground state $\rho$ of $V$ and consequently $\langle W\rangle_{\rho}=0$. More specifically, $W$ must have the ground states of $V$ as its zero solutions. This observation will limit the set of $W$ we can choose from. For example, if $V=a^\dagger a$ is the energy operator of a quantum oscillator, we will not be able to establish $\mathcal{G}(V)\leq-W$ for the position operator $W=(a+a^\dagger)^2$ because the ground state $|0\rangle\langle0|$ of $V$ has nonzero variance in position. In other words, it is impossible to generate states with zero variance in position by stabilizing the energy of the system. This example reveals the fundamental difficulty in stabilizing non-commuting operators, which is also the implication of the Heisenberg uncertainty principle. Nevertheless, through the stability of $V$ we can still infer the information about the non-commuting operators that are restricted in a subspace. In addition, $W$ can also be used to characterize other invariant limit sets of $\rho_t$ besides the set of the ground states of $V$. We illustrate these ideas in the following example:

\begin{exam}\rm
Consider a single-qubit system with energy operator $V=\frac{1}{2}(1+\sigma_z)$. The aim is to make the expectation of the coherence operator $W=\frac{1}{2}(1+\sigma_x)$ zero and in the same time stabilize the energy of the system. However, the non-commuting observables $V$ and $W$ cannot be stabilized simultaneously via $\mathcal{G}(V)\leq-W$ since the ground state $|1\rangle\langle1|$ of $V$ is not the zero solution of $W$. An alternative solution to this problem is to consider the augmented system with an ancillary qubit and define $W=|0\rangle\langle0|\otimes\frac{1}{2}(1+\sigma_{x_2})$. The energy of the two-qubit system is characterized by the operator $V$ as $V=\sigma_{z_1}+\sigma_{z_2}$. $\sigma_{z_i}$ is the Pauli operator $\sigma_z$ acting on the $i$th qubit. The basis of the bipartite system is chosen as the four eigenstates $\{|00\rangle,|01\rangle,|10\rangle,|11\rangle\}$, leading to the following expression of $V$ and $W$
$$
V=\left(
\begin{array}{cccc}
2&0&0&0\\
0&0&0&0\\
0&0&0&0\\
0&0&0&-2
\end{array}
\right),
W=\left(
\begin{array}{cccc}
\frac{1}{2}&\frac{1}{2}&0&0\\
\frac{1}{2}&\frac{1}{2}&0&0\\
0&0&0&0\\
0&0&0&0
\end{array}
\right).
$$
Although $V$ is not positive, Theorem~$5$ still applies to this example by shifting $V$ with a constant. We can engineer $\mathcal{G}(V)$ (See Appendix) through engineering the couplings between the eigenstates. By introducing the couplings $l|01\rangle\langle00|$ and $l|11\rangle\langle01|$ with $|l|^2=\frac{1}{2}$, $\mathcal{G}(V)$ will become
$$
\left(
\begin{array}{cccc}
-1&0&0&0\\
0&-1&0&0\\
0&0&0&0\\
0&0&0&0
\end{array}
\right).
$$
Set the Hamiltonian control $H$ as $-\frac{1}{2}i|00\rangle\langle01|+\frac{1}{2}i|01\rangle\langle00|=|0\rangle\langle0|\otimes\sigma_{y_2}$, the new $\mathcal{G}(V)$ is
$$
\left(
\begin{array}{cccc}
-1&-1&0&0\\
-1&-1&0&0\\
0&0&0&0\\
0&0&0&0
\end{array}
\right)
$$
which satisfies the required inequality $\mathcal{G}(V)\leq-W$. The system will converge to the zero solutions of $W=|0\rangle\langle0|\otimes\frac{1}{2}(1+\sigma_{x_2})$ while the energy operator is stabilized (the energy of the two-qubit system is decreasing).

Given the density matrix of $\rho$ as
$$
\left(
\begin{array}{cccc}
\rho_{00}&\rho_{01}&\rho_{02}&\rho_{03}\\
\rho_{10}&\rho_{11}&\rho_{12}&\rho_{13}\\
\rho_{20}&\rho_{21}&\rho_{22}&\rho_{23}\\
\rho_{30}&\rho_{31}&\rho_{32}&\rho_{33}
\end{array}
\right),
$$
the limit states will satisfy $\langle|0\rangle\langle0|\otimes\frac{1}{2}(1+\sigma_{x_2})\rangle_{\rho}=0$ and hence $\rho_{00}+\rho_{01}+\rho_{10}+\rho_{11}=0$. In this example we are able to infer the information about the coherence $\rho_{01}+\rho_{10}$ between $|00\rangle$ and $|01\rangle$ within the two-level subspace through the generator of the energy operator $V$. Note that $V$ and $W$ do not commute. One particular state satisfying $\rho_{00}+\rho_{01}+\rho_{10}+\rho_{11}=0$ is $\frac{1}{2}(|00\rangle\langle00|-|00\rangle\langle01|-|01\rangle\langle00|+|01\rangle\langle01|)$, which is an invariant state of the system. Note that the space spanned by $\{|10\rangle,|11\rangle\}$ also satisfies $\rho_{00}+\rho_{01}+\rho_{10}+\rho_{11}=0$. We can further narrow down the set of limit points by making $G_{22}$ negative via the methods introduced in the Appendix such that the invariant set will only contain states that are either in the space spanned by $\{|00\rangle,|01\rangle\}$ with stabilized coherence, or in the ground state $|11\rangle\langle11|$.

Moreover, if we make a projection $|0\rangle\langle0|$ on the $1$st qubit via a quantum measurement, the reduced quantum state of the $2$nd qubit will satisfy $\langle\frac{1}{2}(1+\sigma_{x_2})\rangle_{\rho_2}=0$. Interestingly, by stabilizing the energy operator $\sigma_{z_1}+\sigma_{z_2}$ of the augmented system and then making a projective measurement $|0\rangle\langle0|$, we are able to stabilize the coherence operator $\frac{1}{2}(1+\sigma_x)$ of the qubit in the end.

The interpretation of these results is as follows: Extra space is needed to store the excess noises introduced by the Heisenberg uncertainty principle. This idea is similar to the design of non-degenerate parametric amplifier, where additional channel of noise input is introduced in order to amplify the amplitude and phase quadratures simultaneously.
\qed
\end{exam}

For a positive operator $W$ with unbounded $\mathcal{G}(W)$, we have the following theorem.
\begin{thm}
If there exists a Lyapunov operator $V$ and a positive operator $W$ such that
\begin{equation*}
\mathcal{G}(V)\leq-W, \quad\mathcal{G}(W)\leq0,
\end{equation*}
then $\lim_{t\rightarrow\infty}\langle
V(t)\rangle_{\rho_0}=\lim_{t\rightarrow\infty}\langle V\rangle_{\rho_t}$
exists for any state trajectory $\rho_t$ and (\ref{conlasa}) holds.
\end{thm}
\begin{proof}
Following the same reasoning as in the proof of Theorem~$5$, we can conclude $\int_0^\infty\langle W\rangle_{\rho_{t}}dt<+\infty$. The conditions $\mathcal{G}(W)\le0$ and $\langle W\rangle_{\rho_{t}}$ is bounded from below guarantee that $\langle W\rangle_{\rho_{t}}$ is convergent. The limit of $\langle W\rangle_{\rho_{t}}$ can only be $0$ because $\int_0^\infty\langle W\rangle_{\rho_{t}}dt$ is finite.
\end{proof}

\begin{rem}
Suppose $\mathcal{G}(V)\leq-cV$, $c>0$, and $V$ is a Lyapunov operator. Let $W=cV$ and we have $\mathcal{G}(W)=c\mathcal{G}(V)\leq-c^2V\leq0$. The system will converge to the zero solutions of $V$, or equivalently speaking, to the set of ground states $Z_V=\{\rho:\langle V\rangle_\rho=0\}$.
\end{rem}

Theorem~$6$ can be extended to treat a general Hermitian operator $W$
\begin{thm}
If there exists a Lyapunov operator $V$ satisfying $\langle V(t)\rangle\leq c$ for $t>0$ and
\begin{equation*}
\mathcal{G}(V)=W,
\end{equation*}
where the generator of $W$ satisfies $\mathcal{G}(W)\leq0$, then
\begin{equation*}
\lim_{t\rightarrow\infty}\langle W(t)\rangle_{\rho_0}=\lim_{t\rightarrow\infty}\langle W\rangle_{\rho_t}=0.
\end{equation*}
\end{thm}
\begin{proof}
$\langle V(t)\rangle$ is bounded for all $t>0$. From $(\ref{integene1})$ we know that $-\infty<\int_0^\infty\langle W\rangle_{\rho_{t}}dt<+\infty$. If $\mathcal{G}(W)\leq0$ by assumption, then the monotonic sequence $\langle W\rangle_{\rho_{t}}$ is bounded from below and hence will converge to a limit. The limit is exactly $0$ since the integral $\int_0^\infty\langle W\rangle_{\rho_{t}}dt$ is bounded.
\end{proof}

\section{Stability within the Invariant Set}\label{sec7}
We have used multiple Lyapunov conditions in Theorem~$6$ and Theorem~$7$. Similarly, we can use additional Lyapunov conditions to further engineer the dynamics of the trajectories within the invariant set. For example, we can make use of the Lyapunov operator $W=V^2$ to drive the system states to the zero solutions of $V$, where in general the system will converge only to the zero solutions of $\mathcal{G}(V)$ by LaSalle invariance principle.

As we have known from classical stochastic stability and quantum semigroup theory, the asymptotic dynamics of the trajectories are determined by the diffusion terms \cite{Khas} or the dissipation functional $\mathfrak{D}(\cdot)$ \cite{Frigerio78,Alberto1982,fagnola2003quantum}. As shown in the proof of Theorem~$3$, we can make explicit connection between the dissipation functional and the diffusion terms $j_t (\mathcal{B}(X)), j_t (\mathcal{C}(X))$ by calculating $\mathcal{G}(V^2)$.

\begin{thm}
Suppose $\mathcal{G}(V)\leq0$ for the Lyapunov operator $V$ of a finite-dimensional system. The state trajectory $\rho_t$ will converge to the set of zero solutions $Z_V=\{\rho:\langle V\rangle_\rho=0\}$ if $\langle[L^\dagger,V][V,L]\rangle_\rho>0$ for $\rho\notin Z_V$ and $[\mathcal G(V),V]=0$.
\end{thm}
\begin{proof}
Since $\mathcal{G}(V)\leq0$, $\lim_{t\rightarrow\infty}\langle V\rangle_{\rho_t}$ exists and $Z_V$ is an invariant set. We only need to prove that $\rho_t$ will exit the domain $\{\rho:\langle V\rangle_\rho\geq\epsilon\}$ for arbitrary $\epsilon>0$. Consider the positive operator $W(V)=V^2$. Similar to the derivations in Theorem~$3$, the generator for $W(V)$
can be calculated using the quantum It\={o} formula
\begin{equation*}
\mathcal{G}(W)=V\mathcal{G}(V)+\mathcal{G}(V)V+\mathfrak{D}(V)
\end{equation*}
with
\begin{equation*}
\mathfrak{D}(V)=\mathcal{B}(V)^2+ \mathcal{C}(V)^2+i\mathcal{B}(V)\mathcal{C}(V)-i\mathcal{C}(V)\mathcal{B}(V)=[L^\dagger,V][V,L].
\end{equation*}
For finite-dimensional system, any state trajectory is tight. Suppose the trajectory $\rho_t$ is restricted to a domain $\{\rho:\langle V\rangle_\rho\geq\epsilon\}$ for some $\epsilon>0$. Then by Theorem~$1$ there exists an invariant state $\rho_I$ which is the limit point of the tight sequence $\frac{1}{t}\int_0^{t}\rho_{t^{'}}dt^{'}$. Note that $\frac{1}{t}\int_0^{t}\rho_{t^{'}}dt^{'}$ is the mean of the sequence $\rho_t$, so $\rho_I$ is in the same domain $\{\rho:\langle V\rangle_\rho\geq\epsilon\}$ as $\rho_t$ which means $\rho_I\notin Z_V$.

Let the initial state be exactly the invariant state $\rho_I$. First we prove $\langle V\mathcal{G}(V)+\mathcal{G}(V)V\rangle_{\rho_I}=0$. Since $V$ is positive and $[\mathcal G(V),V]=0$, $V\mathcal{G}(V)$ and $\mathcal{G}(V)V$ are negative hermitian operators which make $\langle V\mathcal{G}(V)+\mathcal{G}(V)V\rangle_{\rho_I}\leq0$. Furthermore, we have $\langle V\mathcal{G}(V)+\mathcal{G}(V)V\rangle_{\rho_I}=\langle (V+\beta)\mathcal{G}(V)+\mathcal{G}(V)(V+\beta)\rangle_{\rho_I}$ due to the fact that $\langle\mathcal{G}(V)\rangle_{\rho_I}=0$. $V$ is bounded, so we can choose $\beta<0$ such that $V+\beta$ is negative. Given this $\beta$, we can conclude $\langle (V+\beta)\mathcal{G}(V)+\mathcal{G}(V)(V+\beta)\rangle_{\rho_I}\geq0$ which gives us $\langle V\mathcal{G}(V)+\mathcal{G}(V)V\rangle_{\rho_I}\geq0$. So $\langle V\mathcal{G}(V)+\mathcal{G}(V)V\rangle_{\rho_I}=0$. Next we have the following relation by integrating $\mathcal{G}(W)$
\begin{eqnarray}
\langle W(V)\rangle_{\rho_I}-{\langle W(V)\rangle_{\rho_I}}
&=&\int_0^t\langle V\mathcal{G}(V)+\mathcal{G}(V)V\nonumber +\mathfrak{D}(V)\rangle_{\rho_I}dt^{'}\\
&=& \int_0^t\langle \mathfrak{D}(V)\rangle_{\rho_I}dt^{'}.
\end{eqnarray}
The LHS of the equality is zero, however the RHS of the
equality is strictly positive, since $\langle \mathfrak{D}(V)\rangle_{\rho_I}=\langle[L^\dagger,V][V,L]\rangle_{\rho_I}>0$ by assumption. So we arrive at a contradiction. The contradiction shows that a trajectory $\rho_t$ cannot be confined to the domain $\{\rho:\langle V\rangle_\rho\geq\epsilon\}$. Hence $\rho_t$ will approach $Z_V$ asymptotically.
\end{proof}

\begin{coro}
Assume in a finite-dimensional system the Lyapunov operator $V$ has the decomposition $V=M^\dagger M$. If $M$ solves $M=[V,L]$ and then $\mathcal{G}(V)\leq0,[\mathcal G(V),V]=0$, the state trajectory $\rho_t$ will converge to the zero solutions of $V$.
\end{coro}

\begin{exam}\rm
Consider a qubit with $V=\frac{1}{2}(1+\sigma_z)$, or in matrix expression
$$
V=\left(
\begin{array}{cc}
1&0\\
0&0
\end{array}
\right).
$$
The decomposition is found to be $V=\sigma_+\sigma_-$ with
$$
\sigma_+=\left(
\begin{array}{cc}
0&1\\
0&0
\end{array}
\right),
\sigma_-=\left(
\begin{array}{cc}
0&0\\
1&0
\end{array}
\right).
$$
$\sigma_+^\dagger=\sigma_-=M$. The solution to $M=[V,L]$ is
$$
L=\left(
\begin{array}{cc}
a&0\\
1&b
\end{array}
\right)
$$
with $a$ and $b$ being arbitrary constants. With this $L$, the dissipation part $\mathfrak{L}(V)$ equals
$$
\left(
\begin{array}{cc}
-1&-\frac{b}{2}\\
-\frac{b}{2}&0
\end{array}
\right).
$$
Let $H=0$ and $b=0$, then $\mathcal{G}(V)=\mathfrak{L}(V)\leq0$. The system will converge to the ground state $|1\rangle\langle1|$.

\end{exam}

\section{Conclusion}\label{sec8}
Many theorems concerning asymptotic properties of quantum
Markov semigroups have the existence of a faithful invariant state as an
essential assumption. We have derived sufficient conditions to verify this
assumption. If these sufficient conditions hold, the unique and faithful
state is an equilibrium point which is also globally attractive. Our
approach makes use of the Lyapunov method complemented by additional algebraic
conditions. Our result exhibits some analogy with the classical
Foster-Lyapunov theory concerning the existence of invariant measures of
Markov processes. Beyond invariant states, we have introduced the quantum
invariance principle to characterize the set of limit states of the system
dynamics. More specifically, the system will asymptotically converge to the ground state of an operator $W$ if we are able to engineer the generator of a Lyapunov operator $V$. These invariance theorems are established via a Lyapunov
inequality between these two operators, which has potential to provide useful
tools for stability analysis of the non-commutative algebra associated with
general quantum coherent control systems. Moreover, the system can be
driven further to the ground state of $V$ within the invariant set if additional
conditions on the Lyapunov operator can be engineered. These results may also find essential applications in quantum information processing, since the outcomes of quantum computations can be encoded in the ground state of a particular operator \cite{Verstraete09}.

\section{Appendix}
In this appendix we introduce a constructive method to engineer a negative generator for the Lyapunov operator $V$.

First we will focus on engineering the dissipation part $\mathfrak{L}(V)=\frac{1}{2}(2L^{\dagger}VL-L^{\dagger}LV-VL^{\dagger}L)$ of the generator $\mathcal{G}(V)$ by assuming $[V,H]=0$. In a separable space we can decompose $L$ and $V$ as
$$L=\left(
\begin{array}{cc}
L_{00} & L_{01} \\
L_{10} & L_{11}
\end{array}
\right), V=\left(
\begin{array}{cc}
V_{00} & V_{01} \\
V_{10} & V_{11}
\end{array}
\right).$$$V$ is a positive hermitian operator, so we can always
make $V_{01}=V_{10}=0$ through spectral decomposition. The generator
$\mathcal{G}(V)$ is calculated to be
$$\mathcal{G}=\left(
\begin{array}{cc}
G_{00} & G_{01} \\
G_{10} & G_{11}
\end{array}
\right)$$with
\begin{eqnarray*}
G_{00}&=&(L_{00}^{\dagger}V_{00}L_{00}+L_{10}^{\dagger}V_{11}L_{10})-\frac{1}{2}\{L_{00}^{\dagger}L_{00}+L_{10}^{\dagger}L_{10},V_{00}\},\\
G_{01}&=&(L_{00}^{\dagger}V_{00}L_{01}+L_{10}^{\dagger}V_{11}L_{11})-\frac{1}{2}(L_{00}^{\dagger}L_{01}+L_{10}^{\dagger}L_{11})V_{11}-\frac{1}{2}V_{00}(L_{00}^{\dagger}L_{01}+L_{10}^{\dagger}L_{11}),\\
G_{10}&=&G_{01}^\dagger,\\
G_{11}&=&(L_{01}^{\dagger}V_{00}L_{01}+L_{11}^{\dagger}V_{11}L_{11})-\frac{1}{2}\{V_{11},L_{01}^{\dagger}L_{01}+L_{11}^{\dagger}L_{11}\}.
\end{eqnarray*}
Here we only present the calculations for two-level system. For higher dimensional systems, the blocks of $G,L$ and $H$ will be
matrices or operators. However, we can still do analysis for arbitrary dimensional systems by carefully engineering the two-dimensional subsystems and compensating the interactions between different subsystems. This approach is possible because of the
linearity of the generator $\mathcal{G}$. For example, $G_{00}$ can be divide into
the internal dynamics $2L_{00}^{\dagger}V_{00}L_{00}-\{L_{00}^{\dagger}L_{00},V_{00}\}$
and the interaction with other dimensions $2L_{10}^{\dagger}V_{11}L_{10}-\{L_{10}^{\dagger}L_{10},V_{00}\}$.

\subsection{$V_{00}=V_{11}$}
We can set $L_{00}=0$ and $L_{11}=0$ to make
$G_{01}=G_{10}=0$. However, due to the degeneracy of $V$, we also
have $G_{00}=G_{11}=0$. Therefore, $\mathcal{G}=0$ and the entire
two-dimensional space is irreducible. The off-diagonal elements of $L$ will not affect $\mathcal{G}$.

\subsection{$V_{00}{\neq}V_{11}$}
Without loss of generality we assume $V_{11}-V_{00}=v>0$. Again assuming $L_{00}=0$ and $L_{11}=0$, the
generator becomes
$$\mathcal{G}=\left(
\begin{array}{cc}
vL_{10}^{\dagger}L_{10}&0\\
0&-vL_{01}^{\dagger}L_{01}
\end{array}
\right).$$Now we can set $L_{10}=0,L_{01}=l\neq0$ so that
$G_{00}=0,G_{11}<0$. The coupling operator $L$ for engineering negative $\mathcal{G}(V)$ with non-degenerate spectrum could be
$$L=\left(
\begin{array}{cc}
0&l\\
0&0
\end{array}
\right).$$.

\subsection{$[V,H]\neq0$}
$[V,H]\neq0$ happens when $V$ is not representing the energy of the
system or additional Hamiltonian control $H_c$ is needed for
stabilization. Since $V_{01}=V_{10}=0$, the commutator $C=-i[V,H]$
can be calculated as
$$C=-i\left(
\begin{array}{cc}
[V_{00},H_{00}]&V_{00}H_{01}-H_{01}V_{11}\\
V_{11}H_{10}-H_{10}V_{00}&[V_{11},H_{11}]
\end{array}\right).$$
$C_{00}$ is the internal unitary dynamics within the subspace
$X_{00}$. For two-dimensional system, $V_{ij}$ and $H_{ij}$ are
complex numbers, which gives $C_{00}=C_{11}=0$. If
$V_{00}=V_{11}$, then $C_{01}=C_{10}=0$ and the unitary dynamics
induced by $H$ will not affect $\mathcal{G}$. If $V_{00}\neq V_{11}$,
we set $V_{11}-V_{00}=v>0$ and $L_{10}=0$. In this case, the
generator will still satisfy the relations $G_{00}=0$ and $G_{11}<0$
after adding $C$ to $\mathcal{G}$. However, $G_{01}$ cannot be made
vanish if the diagonal entries of $L$ are all zero. In fact, we have
\begin{equation*}
G_{01}=L_{00}^\dagger L_{01}(V_{00}-V_{11})-iH_{01}(V_{00}-V_{11}),
\end{equation*}
so $L_{00}^\dagger L_{01}=iH_{01}$ must be satisfied. If we choose
$L_{01}=l\neq0$, the coupling operator $L$ should be in the
following form
$$L=\left(
\begin{array}{cc}
-\frac{iH_{01}^*}{l^*}&l\\
0&L_{11}
\end{array}\right)
$$ to completely eliminate the influence of $H$.

\bibliographystyle{plain}

\end{document}